\PassOptionsToPackage{unicode}{hyperref}
\PassOptionsToPackage{hyphens}{url}
\PassOptionsToPackage{dvipsnames,svgnames,x11names}{xcolor}
\documentclass[
  11pt,
  a4paper,
  DIV=11,
  numbers=noendperiod]{scrartcl}

\usepackage{amsmath,amssymb}
\usepackage{iftex}
\ifPDFTeX
  \usepackage[T1]{fontenc}
  \usepackage[utf8]{inputenc}
  \usepackage{textcomp} 
\else 
  \usepackage{unicode-math}
  \defaultfontfeatures{Scale=MatchLowercase}
  \defaultfontfeatures[\rmfamily]{Ligatures=TeX,Scale=1}
\fi
\usepackage{lmodern}
\ifPDFTeX\else  
\fi
\IfFileExists{upquote.sty}{\usepackage{upquote}}{}
\IfFileExists{microtype.sty}{
  \usepackage[]{microtype}
  \UseMicrotypeSet[protrusion]{basicmath} 
}{}
\makeatletter
\@ifundefined{KOMAClassName}{
  \IfFileExists{parskip.sty}{%
    \usepackage{parskip}
  }{
    \setlength{\parindent}{0pt}
    \setlength{\parskip}{6pt plus 2pt minus 1pt}}
}{
  \KOMAoptions{parskip=half}}
\makeatother
\usepackage{xcolor}
\usepackage{longtable,booktabs,array}
\usepackage{calc} 
\usepackage{etoolbox}
\makeatletter
\patchcmd\longtable{\par}{\if@noskipsec\mbox{}\fi\par}{}{}
\makeatother
\IfFileExists{footnotehyper.sty}{\usepackage{footnotehyper}}{\usepackage{footnote}}
\makesavenoteenv{longtable}
\usepackage{graphicx}
\makeatletter
\def\maxwidth{\ifdim\Gin@nat@width>\linewidth\linewidth\else\Gin@nat@width\fi}
\def\maxheight{\ifdim\Gin@nat@height>\textheight\textheight\else\Gin@nat@height\fi}
\makeatother
\setkeys{Gin}{width=\maxwidth,height=\maxheight,keepaspectratio}
\makeatletter
\def\fps@figure{htbp}
\makeatother
\newlength{\cslhangindent}
\newlength{\csllabelwidth}
\newlength{\cslentryspacingunit} 
\newenvironment{CSLReferences}[2] 
 {
  \setlength{\parindent}{0pt}
  \ifodd #1
  \let\oldpar\par
  \def\par{\hangindent=\cslhangindent\oldpar}
  \fi
  \setlength{\parskip}{#2\cslentryspacingunit}
 }%
 {}
\usepackage{calc}

\KOMAoption{captions}{tableheading}
\makeatletter
\@ifpackageloaded{caption}{}{\usepackage{caption}}
\AtBeginDocument{%
\ifdefined\contentsname
  \renewcommand*\contentsname{Table of contents}
\else
  \newcommand\contentsname{Table of contents}
\fi
\ifdefined\listfigurename
  \renewcommand*\listfigurename{List of Figures}
\else
  \newcommand\listfigurename{List of Figures}
\fi
\ifdefined\listtablename
  \renewcommand*\listtablename{List of Tables}
\else
  \newcommand\listtablename{List of Tables}
\fi
\ifdefined\figurename
  \renewcommand*\figurename{Figure}
\else
  \newcommand\figurename{Figure}
\fi
\ifdefined\tablename
  \renewcommand*\tablename{Table}
\else
  \newcommand\tablename{Table}
\fi
}
\@ifpackageloaded{float}{}{\usepackage{float}}
\floatstyle{ruled}
\@ifundefined{c@chapter}{\newfloat{codelisting}{h}{lop}}{\newfloat{codelisting}{h}{lop}[chapter]}
\floatname{codelisting}{Listing}

\usepackage{amsthm}
\theoremstyle{plain}
\newtheorem{theorem}{Theorem}[section]
\theoremstyle{plain}
\newtheorem{lemma}{Lemma}[section]
\theoremstyle{definition}
\newtheorem{definition}{Definition}[section]
\theoremstyle{remark}
\AtBeginDocument{}

\makeatother
\makeatletter
\@ifpackageloaded{caption}{}{\usepackage{caption}}
\@ifpackageloaded{subcaption}{}{\usepackage{subcaption}}
\makeatother
\makeatletter
\@ifpackageloaded{tcolorbox}{}{\usepackage[skins,breakable]{tcolorbox}}
\makeatother
\makeatletter
\@ifundefined{shadecolor}{\definecolor{shadecolor}{rgb}{.97, .97, .97}}
\makeatother
\ifLuaTeX
\usepackage[bidi=basic]{babel}
\else
\usepackage[bidi=default]{babel}
\fi
\babelprovide[main,import]{british}

\def\languageshorthands#1{}
\ifLuaTeX
  \usepackage{selnolig}  
\fi
\IfFileExists{bookmark.sty}{\usepackage{bookmark}}{\usepackage{hyperref}}
\IfFileExists{xurl.sty}{\usepackage{xurl}}{} 
\hypersetup{
  pdftitle={PageRank and the Bradley--Terry model},
  pdfauthor={David Antony Selby},
  pdflang={en-GB},
  pdfkeywords={PageRank, Bradley--Terry model, network
centrality, bibliometrics, rankings},
  colorlinks=true,
  linkcolor={blue},
  filecolor={Maroon},
  citecolor={Blue},
  urlcolor={Blue},
  pdfcreator={LaTeX via pandoc}}

\title{PageRank and the Bradley--Terry model\thanks{This work is based
on a chapter from the author's unpublished PhD thesis
(\protect\hyperlink{ref-selby2020}{Selby 2020}).}}
\author{\href{mailto:David_Antony.Selby@dfki.de}{David Antony Selby}}
\date{\small Data Science and its Applications\\German Research Centre for Artificial Intelligence (DFKI)\\Kaiserslautern, Germany}

\begin{document}
\maketitle
\begin{abstract}
\noindent PageRank and the Bradley--Terry model are competing approaches to ranking entities such as teams in sports tournaments or journals in citation networks. The Bradley--Terry model is a classical statistical method for ranking based on paired comparisons. The PageRank algorithm ranks nodes according to their importance in a network. Whereas Bradley--Terry scores are computed via maximum likelihood estimation, PageRanks are derived from the stationary distribution of a Markov chain. More recent work has shown maximum likelihood estimates for the Bradley--Terry model may be approximated from such a limiting distribution, an interesting connection that has been discovered and rediscovered over the decades. Here we show---through relatively simple mathematics---a connection between paired comparisons and PageRank that exploits the quasi-symmetry property of the Bradley--Terry model. This motivates a novel interpretation of Bradley--Terry scores as `scaled' PageRanks, and vice versa, with direct implications for citation-based journal ranking metrics.
\end{abstract}
\ifdefined\Shaded\renewenvironment{Shaded}{\begin{tcolorbox}[sharp corners, interior hidden, boxrule=0pt, borderline west={3pt}{0pt}{shadecolor}, enhanced, frame hidden, breakable]}{\end{tcolorbox}}\fi

\hypertarget{introduction}{%
\section{Introduction}\label{introduction}}

This paper discusses two alternative quantitative approaches for ranking
players in tournaments, or entities in social networks: the
Bradley--Terry paired comparisons model and the PageRank centrality
score. Both of these methods offer advantages over simpler metrics based
on unweighted counts of wins and losses.

By examining the theory underpinning each method---one a generalized
linear model and the other based on a Markov chain---it can be shown
that the two are closely connected. Under idealized conditions, a
modified PageRank metric---proposed, among others, by Pinski and Narin
(\protect\hyperlink{ref-pinski1976}{1976}) but seemingly overlooked in
recent years---yields rankings exactly equal to those from the
Bradley--Terry model. We present a novel proof for this using a
quasi-symmetry representation, and use the delta method to demonstrate
some special cases where a PageRank-based score is an asymptotically
efficient estimator for the Bradley--Terry model.

\hypertarget{network-rankings}{%
\section{Network rankings}\label{network-rankings}}

\hypertarget{pagerank}{%
\subsection{PageRank}\label{pagerank}}

PageRank, named for Larry Page, was developed in the 1990s by Google for
their search engine (\protect\hyperlink{ref-page1999}{Page et al.
1999}). More recently this approach has been applied to citation
networks in the form of the Eigenfactor and SCImago Journal Rank metrics
(\protect\hyperlink{ref-bergstrom2007}{Bergstrom 2007};
\protect\hyperlink{ref-falagas2008}{Falagas et al. 2008}).

The Markov chain corresponding to PageRank has a simple analogy, of a
random walk around the graph. Consider an imaginary PhD student, who
opens a random journal from the library. The student selects at random a
reference from within that journal and proceeds to read the cited
journal. Then a third journal is selected from the references of the
second, a fourth from the third, and so on. The proportion of overall
time spent reading a particular journal may be seen as a measure of that
journal's importance (\protect\hyperlink{ref-bergstrom2007}{Bergstrom
2007}).

Google's innovation\footnote{Strictly speaking, the damping factor was
  introduced as part of Katz centrality
  (\protect\hyperlink{ref-katz1953}{1953}). See Vigna
  (\protect\hyperlink{ref-vigna2016}{2016}).} was the addition of a
damping factor, \(\alpha\). At any time, with probability \(1-\alpha\),
our PhD student gets bored, returns their current journal to the shelf
and selects a new journal at random from the library. In a random walk
this would be the ability to `teleport' randomly from one node to
another. This helps link up disconnected components of sparse graphs,
making the process ergodic. PageRank without teleportation
(\(\alpha=1\)) is called \emph{undamped}.

Let the matrix \(C\) describe the results of a tournament, where
\(c_{ij}\) is the number of times player \(i\) beats \(j\) (in a
directed network, the number or weight of links from node \(j\) to
\(i\)) for \(i,j=1,\dots,n\). Let \(A\) be a diagonal matrix equal to
the column sums of \(C\). Then PageRank is the stationary distribution
of the Markov chain with probability transition matrix

\[
P_\alpha = \alpha C A^{-1} + \frac{1-\alpha}n e e^T,
\]

where \(e\) is an \(n\)-vector of ones. We can compute PageRank from the
leading right eigenvector, \(\pi = P\pi\). Undamped PageRank is the
leading eigenvector of \(P_1 = CA^{-1}\).

\hypertarget{the-bradleyterry-model}{%
\subsection{The Bradley--Terry model}\label{the-bradleyterry-model}}

An alternative approach for ranking players in a tournament is a paired
comparisons model (\protect\hyperlink{ref-bradley1952}{Bradley and Terry
1952}). Assume that in every pairing, \(i\) defeats \(j\) or vice versa.
The log odds are given by

\begin{equation}\protect\hypertarget{eq-bt}{}{ \operatorname{log\,odds} (i~\text{beats}~j \mid i~\text{and}~j~\text{compete}) = \mu_i - \mu_j }\label{eq-bt}\end{equation}

where \(\mu_i\) and \(\mu_j\) are the respective \emph{ability scores}
of players \(i\) and \(j\). Players can be ranked on a linear scale
according to their scores, with victories against higher-ability players
contributing more against `easy' opponents. Extensions of the
Bradley--Terry model allow for tied games, `home advantage' and other
contextual effects, with estimates for the scores computed via maximum
likelihood estimation using standard statistical software
(\protect\hyperlink{ref-turner2012}{Turner and Firth 2012}).

Like PageRank, the Bradley--Terry model has been applied to citation
networks: Stigler (\protect\hyperlink{ref-stigler1994}{1994}) proposed a
model that measured `export scores' for academic journals. The analogy
corresponds to bilateral trade; intellectual influence being `exported'
(i.e.~citations received) and `imported' (citations given) among trading
partners (journals). Larger export scores imply greater influence.
According to Stigler, desirable properties for a measure of influence
include insensitivity to self-links (not applicable to football matches,
but an important consideration in journal rankings), absence of size
bias, demonstrated through invariance to node aggregation (merging of
journals).

One problem with such a model, however, is that it assumes that outcomes
of paired comparisons are independent, which can lead to overdispersion.
Varin, Cattelan, and Firth (\protect\hyperlink{ref-varin2016}{2016})
describe the use of quasi-likelihood estimation to fit a `quasi-Stigler'
model with an additional parameter of dispersion.

While PageRank and the Bradley--Terry model are both based on
well-defined stochastic principles, PageRank is not strictly a `model'
in the statistical sense, but rather a unique characteristic of a
matrix: an exact measure of the `centrality' of the vertices in a graph.
There is no `ground truth' against which to compare the computed score
vector. Contrastingly, the paired comparisons model allows for
uncertainty quantification, which can be conveniently visualized using
\emph{quasi-variances} (\protect\hyperlink{ref-firth2004}{Firth 2004})
and investigation of lack of fit through analysis of \emph{journal
residuals} (\protect\hyperlink{ref-varin2016}{Varin, Cattelan, and Firth
2016}).

\hypertarget{bibliometric-connections}{%
\section{Bibliometric connections}\label{bibliometric-connections}}

\hypertarget{pagerank-precursors}{%
\subsection{PageRank precursors}\label{pagerank-precursors}}

As Franceschet (\protect\hyperlink{ref-franceschet2011}{2011}) and Vigna
(\protect\hyperlink{ref-vigna2016}{2016}) note, the ideas behind
PageRank---and eigenvector centrality or `spectral ranking' more
generally---stem from work originating much earlier in the twentieth
century. One of these earlier proposed applications for such an score
was in ranking scholarly journals, and thus the Eigenfactor metrics
(\protect\hyperlink{ref-bergstrom2007}{Bergstrom 2007}) might be
considered a \emph{re}-adaptation of eigenvector centrality to
bibliometrics following its relatively recent adoption for search
engines. The notion of a damping factor can also be attributed to Katz
(\protect\hyperlink{ref-katz1953}{1953}).

Pinski and Narin (\protect\hyperlink{ref-pinski1976}{1976}) proposed
three related measures for citation influence: \emph{total influence},
\emph{influence weight} and \emph{influence per publication}. A citation
network, like the Web, is a directed, weighted graph representing nodes
(journals or web pages) and the asymmetric relationships between them,
so it is natural that metrics proposed for ranking web sites may also be
applicable to ranking academic publications.

Influence weight (or influence \emph{per outgoing reference}), is
defined by the recursive equation

\begin{equation}\protect\hypertarget{eq-infwt}{}{ w_i = \frac{\sum_{j=1}^n w_j c_{ij}}{\sum_{j=1}^n c_{ji}} }\label{eq-infwt}\end{equation}

for \(i = 1,\dots,n\). Multiplying Equation~\ref{eq-infwt} by the
out-degree of node \(i\) (the size of that journal's bibliography),
\(c_{\cdot i} = \sum_j c_{ji}\), yields the definition of total
influence

\begin{equation}\protect\hypertarget{eq-totinf}{}{ w_i^\text{total} = w_i c_{\cdot i} = \sum_{j=1}^n w_j c_{ij}, }\label{eq-totinf}\end{equation}

and dividing this in turn by \(a_i\), the number of articles published
in journal \(i\), we get the (total) influence per publication

\begin{equation}\protect\hypertarget{eq-infpub}{}{ w_i^\text{pub} = \frac{w_i}{a_i} c_{\cdot i}. }\label{eq-infpub}\end{equation}

Expanding Equation~\ref{eq-totinf}, for every \(i=1,\dots,n\) we have

\[ w_i^\text{total} = \sum_{j=1}^n \frac{c_{ij}}{c_{\cdot j}} w_j^\text{total}, \]

which can be rewritten in matrix notation as
\(w^\text{total} = C A^{-1} {w}^\text{total},\) showing that total
influence is equivalent to undamped PageRank. Indeed, Geller
(\protect\hyperlink{ref-geller1978}{1978}) showed that total influence
(Equation~\ref{eq-totinf}) is the stationary distribution of a Markov
chain. The connections between PageRank and the methods of Pinski and
Narin (\protect\hyperlink{ref-pinski1976}{1976}) as well as similar
independently-proposed metrics, are highlighted in Franceschet
(\protect\hyperlink{ref-franceschet2011}{2011}).

Similarly, influence per publication (Equation~\ref{eq-infpub})
corresponds to an undamped version of Article Influence score, which is
Eigenfactor per publication. One of the motivations for such an
adaptation is to account for a journal size bias, wherein larger
periodicals, containing more articles, accrue more citations simply by
virtue of having more possible things to cite. This does not, however,
solve the problem of review journals, whose articles are longer, thus
also presenting more content as a `target' for citations
(\protect\hyperlink{ref-west2010}{West 2010}).

Palacios-Huerta and Volij
(\protect\hyperlink{ref-palacios-huerta2002}{2002}) named five axioms:
anonymity, weak homogeneity, weak consistency, invariance to citation
intensity and invariance to splitting of journals, arriving at an
approach they describe as the \emph{invariant method}, which is similar
to Pinski and Narin's influence per publication metric, but controls for
citation intensity by dividing by the average bibliography length of
each article. We note these desirable properties closely correlate with
those highlighted by Stigler (\protect\hyperlink{ref-stigler1994}{1994})
for his export scores model.

\hypertarget{theoretical-connection}{%
\section{Theoretical connection}\label{theoretical-connection}}

\hypertarget{quasi-symmetry}{%
\subsection{Quasi-symmetry}\label{quasi-symmetry}}

In this section we introduce the concept of quasi-symmetry, which we can
show underpins the connection between PageRank and the Bradley--Terry
model. Quasi-symmetry was originally defined by Caussinus
(\protect\hyperlink{ref-caussinus1965}{1965}) and is a generalization of
matrix symmetry.

\begin{definition}[]\protect\hypertarget{def-quasisymmetry}{}\label{def-quasisymmetry}

A square matrix \(Q=(q_{ij})_{n \times n}\) is called
\emph{quasi-symmetric} if its elements can be expressed in the form
\(q_{ij} = a_i b_{j} x_{ij}\), where \(x_{ij} = x_{ji}\) for
\(i,j=1,\dots,n\). In matrix notation, this decomposition is written
\(Q = AXB\), where \(A\) and \(B\) are diagonal matrices and \(X\) is
symmetric (\protect\hyperlink{ref-caussinus1965}{Caussinus 1965}).

\end{definition}

\begin{lemma}[]\protect\hypertarget{lem-quasisymmetry}{}\label{lem-quasisymmetry}

A (strictly positive) matrix \(Q\) is quasi-symmetric if and only if it
can be expressed as the product of a diagonal and a symmetric matrix
(\protect\hyperlink{ref-sharp2000}{Sharp and Markham 2000}).

\end{lemma}

\begin{proof}

\[ Q = AXB = AB^{-1}BXB = (AB^{-1})(BXB) = DS, \] where \(D\) is
diagonal and \(S\) is symmetric. Similarly,
\[ Q = AXB = AX(AA^{-1})B = (AXA) (A^{-1}B) = S'D' \] where \(D'\) is
diagonal and \(S'\) symmetric.

\end{proof}

The quasi-symmetry model is a log-linear model for square contingency
tables wherein the expected counts are quasi-symmetric. The
Bradley--Terry model is a logistic formulation of the quasi-symmetry
model (\protect\hyperlink{ref-agresti2013}{Agresti 2013, chap. 10}).
Thus if a results matrix \(C\) is exactly quasi-symmetric, then the
Bradley--Terry model would fit perfectly, with ability scores given by
(the logarithm of) the diagonal elements of \(D\).

\begin{theorem}[]\protect\hypertarget{thm-main}{}\label{thm-main}

Let \(Q=DS\) be a quasi-symmetric matrix with \(D\) diagonal and \(S\)
symmetric. Let \(e\) denote a vector of ones. Let
\(A=\operatorname{diag}(e^TQ)\) be a diagonal matrix with elements equal
to the column sums of \(Q\) and let \(d\) be the vector corresponding to
the diagonal elements of \(D\). Then \(d\) is the leading eigenvector of
\(A^{-1}Q\).

\end{theorem}

This theorem implies that, under quasi-symmetry, the Bradley--Terry
model is equivalent to a scaled version of PageRank. Proving the theorem
requires a few intermediate results.

\begin{lemma}[]\protect\hypertarget{lem-scrooge-eval}{}\label{lem-scrooge-eval}

Let \(C\) be a square matrix and let \(A\) be the diagonal matrix of its
column sums. Then the scaled matrix \(A^{-1}C\) has largest eigenvalue
equal to 1.

\end{lemma}

Stochastic matrices, such as \(CA^{-1}\), always have leading eigenvalue
equal to 1 (\protect\hyperlink{ref-meyer2000}{Meyer 2000}), with
multiplicity 1 if irreducible (by the Perron--Frobenius theorem, see
\protect\hyperlink{ref-palacios-huerta2002}{Palacios-Huerta and Volij
2002}). While the scaled matrix \(A^{-1}C\) is not itself stochastic, it
is \emph{similar} to \(CA^{-1}\).

\begin{lemma}[]\protect\hypertarget{lem-similarity}{}\label{lem-similarity}

Two matrices \(M\) and \(M'\) are called \emph{similar} if there exists
an invertible matrix \(X\) such that \(M' = X^{-1}MX\). Similar matrices
have the same set of eigenvalues
(\protect\hyperlink{ref-newman2010}{Newman 2010, 138}).

\end{lemma}

This is a standard result, but a proof is included for completeness.

\begin{proof}

We have \(M' = X^{-1}MX\) and \(M=XM'X^{-1}\). So

\[
Mv = \lambda v = XM'X^{-1}v \implies M' X^{-1}v = \lambda X^{-1} v.
\]

Hence, if \(\lambda\) is an eigenvalue of \(M\) with eigenvector \(v\),
then \(\lambda\) is also an eigenvalue of \(M'\) with eigenvector
\(X^{-1}v\).

\end{proof}

Now we prove Lemma~\ref{lem-scrooge-eval}.

\begin{proof}

Clearly, \(CA^{-1}\) and \(A^{-1}C\) are similar, because
\(A^{-1}(CA^{-1})A = A^{-1}C,\) so they have the same eigenvalues.
Hence, the leading eigenvalue of \(A^{-1}C\) is equal to 1
(\protect\hyperlink{ref-geller1978}{Geller 1978}).

\end{proof}

Using these existing results, we can now prove Theorem~\ref{thm-main}.

\begin{proof}

\[ \begin{aligned} A^{-1}Cd &= A^{-1} (DS)(De) \\ &= A^{-1}D(e^TDS)^T \\ &= DA^{-1}Ae \\ &= De \\ &= d. \end{aligned} \]
Thus \(d\) is an eigenvector of \(A^{-1}C\) corresponding to eigenvalue
1. By Lemma~\ref{lem-scrooge-eval}, this is the leading eigenvector.

\end{proof}

The proof of Lemma~\ref{lem-similarity} highlights another useful
result: eigenvectors of similar matrices are linear transformations of
one another. Undamped PageRank is the leading eigenvector
\(\pi = CA^{-1}\pi\); influence weight, \(d\) is the leading eigenvector
of the similar but `wrongly scaled' non-stochastic matrix \(A^{-1}C\).

It follows that \(d = A^{-1}\pi\), so we get two ranking metrics for the
price of one eigenvector computation. In other words, we can calculate
Pinski and Narin's influence weight from pre-computed PageRank (and
\emph{vice versa}) even if we do not have access to the original counts
matrix \(C\) but know the vector of marginal out-degrees, harking back
to the definition of influence weight per reference
(Equation~\ref{eq-infwt}).

\hypertarget{reversibility}{%
\subsection{Reversibility}\label{reversibility}}

From the definition of quasi-symmetry it is easy to show that a
quasi-symmetric matrix satisfies the property

\[
    c_{ij} c_{jk} c_{ki} = c_{ji} c_{kj} c_{ik},
\] for each triplet \(i,j,k=1,\dots,n\)
(\protect\hyperlink{ref-caussinus1965}{Caussinus 1965};
\protect\hyperlink{ref-sharp2000}{Sharp and Markham 2000}).

PageRank, meanwhile, represents the stationary distribution of a
(discrete-time) Markov chain. Given a (column-stochastic) transition
matrix \(P\), there exists a stationary distribution
\(\pi = (\pi_i,\dots,\pi_n)\) if \emph{global balance}, \[
    \pi_i \sum_{j=1}^n p_{ji} = \sum_{j=1}^n \pi_j p_{ij},
\] is satisfied for all \(i=1,\dots,n\). The \emph{detailed balance}
equations (also known as local balance),
\begin{equation}\protect\hypertarget{eq-detailed}{}{
    \pi_i p_{ji} = \pi_j p_{ji},
}\label{eq-detailed}\end{equation} hold for all pairs \(i,j=1,\dots,n\)
if and only if the Markov chain is reversible. Kolmogorov's criterion,
which is easy to verify in Equation~\ref{eq-detailed}, is that \[
    p_{ij}p_{jk}p_{ki} = p_{ji}p_{kj}p_{ik}
\] for every triplet \((i,j,k)\) if and only if the Markov chain is
reversible (\protect\hyperlink{ref-kelly1979}{Kelly 1979, chap. 1}). It
follows that a Markov chain is reversible if and only if its probability
transition matrix is quasi-symmetric
(\protect\hyperlink{ref-mccullagh1982}{McCullagh 1982};
\protect\hyperlink{ref-bof2017}{Bof, Baggio, and Zampieri 2017}). The
probability transition matrix for PageRank of a quasi-symmetric
tournament is itself quasi-symmetric, because

\[
CA^{-1} = DSA^{-1} = AA^{-1}DSA^{-1} = AD(A^{-1}SA^{-1})
\]

by commutativity of diagonal matrices. Thus, if (and only if)
Bradley--Terry scores are equal to influence weights, then the
corresponding Markov chain is reversible.

In the linear algebra literature, quasi-symmetric matrices are called
\emph{symmetrizable} and the matrix \(D^{-1}\) is called the
\emph{symmetrizer} (\protect\hyperlink{ref-dias2016}{Dias, Castonguay,
and Dourado 2016}). Symmetrizability, reversibility and quasi-symmetry
are all equivalent.

Unfortunately, a Markov chain with teleportation is not reversible, and
therefore the connection breaks down when \(\alpha < 1\). The
conventional choice of \(\alpha = 0.85\) is somewhat arbitrary
(\protect\hyperlink{ref-newman2010}{Newman 2010, sec. 7.4}) but one
might nonetheless seek the `nearest' reversible Markov chain to an
irreversible one; see Avrachenkov, Ribeiro, and Towsley
(\protect\hyperlink{ref-avrachenkov2010}{2010}), Nielsen and Weber
(\protect\hyperlink{ref-nielsen2015}{2015}) or Tahata
(\protect\hyperlink{ref-tahata2022}{2022}).

\hypertarget{sec-asymptotic}{%
\subsection{Asymptotic efficiency}\label{sec-asymptotic}}

For certain special structures, it can be shown via the delta method
that the (log) influence weight is an asymptotically efficient estimator
for Bradley--Terry scores.

Consider a `round robin' network (tournament) in which every entity
cites (beats) every other entity an equal number of times, \(k\). Then
the citation matrix is \({C} = k {e e}^T\). The corresponding
column-stochastic probability transition matrix is then
\({P} = {C A}^{-1} = \frac1n {e e}^T\).

We model perturbations of this arrangement by
\({C}_t = {C} + t {F}_{ij},\) where \(t\) is a parameter for the
magnitude of perturbation and \({F}_{ij}\) is the \(n \times n\) matrix
with element \((i, j)\) equal to \(1\), element \((j, i)\) equal to
\(-1\) and all other elements equal to zero. When \(t = 0\) then
\({C}_t = {C}_0 = {C}\), so the citation matrix is unperturbed.

The derivative of \({P}\) with respect to \(t\), for a perturbation of
element \((i,j)\), is

\[\frac {\partial{   P}} {\partial t}=\begin{bmatrix}{   0} & {   0} & \cdots &\underbrace {\frac1{n^2}{   e} -\frac1n{   e}_j} _ {\text{column } i}& \cdots &\underbrace {-\frac1{n^2}{   e} + \frac1n{   e}_i} _ {\text{column } j}& \cdots & {   0}\end{bmatrix},\]where
\({ 0}\) is an \(n\)-vector of zeros and \({ e}_i\) is the \(n\)-vector
whose \(i\)th component is \textbackslash(1\textbackslash) with all
other components zero. In other words, the \(i\)th column of
\(\partial{ P}/\partial t\) is all \(\frac1{n^2}\), except element
\((j,i)\), which is equal to \(-\frac{n-1}{n^2}\). The \(j\)th column is
all \(-\frac1{n^2}\), except element \((i,j)\), which is equal to
\(\frac{n-1}{n^2}\). Every other column is filled with zeros.

The partial derivative of an eigenvector is
\(\dot{ \pi}_0 = \pi_0 \dot{{ P}}^T ( {I-P})^\dagger,\) where a dot
denotes partial differentiation with respect to \(t\) and \(^\dagger\)
denotes the Moore--Penrose pseudoinverse
(\protect\hyperlink{ref-golub1986}{Golub and Meyer, Jr. 1986}). Putting
it all together,

\[
\dot{ \pi}_0=\begin{bmatrix}0 &\cdots &\underbrace{\frac1{kn^2}} _ {\text{element } i} &\cdots &\underbrace{\frac{-1}{kn^2}} _ {\text{element } j} &\cdots &0\end{bmatrix}^T.
\]

By the product rule, the derivative of the unnormalized influence weight
is
\(\dot{\operatorname{IW}}_\text{un} = \dot{ D}^{-1} \pi + { D}^{-1} \dot{ \pi}\).
With normalization,
\(\operatorname{IW}_\text{norm} = \frac{{ D}^{-1} \pi}{{ e}^T { D}^{-1} \pi}\).

By the quotient rule

\[
\dot{\operatorname{IW}}_\text{norm}=\frac{\psi \dot{\operatorname{IW}}_\text{un} - \operatorname{IW}_\text{un} \dot\psi}{\psi^2},
\]

where \(\psi = {e}^T \operatorname{IW}_\text{un} = \frac1{kn}\) denotes
the sum of the unnormalized influence weight vector. (Its derivative is
zero.) This yields

\[
\dot{\operatorname{IW}}_\text{norm}(0)=\begin{bmatrix}0 &\cdots &\underbrace{\frac2{kn^2}} _ {\text{element } i} &\cdots &\underbrace{\frac{-2}{kn^2}} _ {\text{element } j} &\cdots &0\end{bmatrix}^T.
\]

The derivative of the log-influence weight, by the chain rule, is

\[
\begin{aligned}\frac{\partial}{\partial t}{\log\operatorname{IW}}(0)&=\frac{\dot{\operatorname{IW}}_\text{norm}}{\operatorname{IW}_\text{norm}} =n\dot{\operatorname{IW}}_\text{norm} \\&=\begin{bmatrix}0 &\cdots &\underbrace{\frac2{kn}} _ {\text{element } i} &\cdots &\underbrace{\frac{-2}{kn}} _ {\text{element } j} &\cdots &0\end{bmatrix}^T.\end{aligned}
\]

Extending this scenario from scalar perturbations, \(t\), to every
possible combination of perturbations of the upper triangle\footnote{N.B.
  any perturbation of the lower triangle of \({ C}\) is equivalent to a
  perturbation to the upper triangle with opposite sign.} of \({ C}\),
we introduce the \(\binom{n}{2}\)-length perturbation vector \({ t}\)
and calculate the Jacobian with respect to the same. The Jacobian of
\(\log\operatorname{IW}\) is an \(n \times \binom{n}{2}\) matrix

\[
{   J} = \begin{bmatrix}+ & + & + & 0 & 0 & - & \cdots & \\- & 0 & 0 & + & + & + & \cdots & \\0 & - & 0 & - & 0 & 0 & \cdots & \\0 & 0 & - & 0 & - & 0 & \cdots &\end{bmatrix},
\]

where `\(+\)' and `\(-\)' represent the positive and negative elements
of the partial derivative
\(\frac{\partial}{\partial t}{\log\operatorname{IW}}\) (that is,
\(\pm\frac2{nk}\)) and where each column corresponds to a perturbation
of the upper triangle of \({ C}\).

If we assume the data are generated from independent binomials (with
\(2k\) trials for each pairing and success probability 1/2), the
covariance matrix \(\Sigma\) is an \(\binom{n}{2}\times\binom{n}{2}\)
diagonal matrix with every diagonal element equal to
\(2k \times \frac12 \times \frac12 = \frac{k}2\).

Applying the delta method, the resulting first-order approximate
covariance matrix for the log-influence weight is
\(\Sigma_{\operatorname{IW}}={J} \Sigma{J}^T\), with elements

\[
\left[ \Sigma_{\operatorname{IW}}\right]_{ij} =\begin{cases}\dfrac{2(n-1)}{kn^2} & i=j, \\[1em]\dfrac{-2}{kn^2}    & i\neq j,\end{cases}
\]

for \(i,j=1,2,\dots,n\).

This is exactly equal to the asymptotic covariance matrix of log-ability
scores from a Bradley--Terry (Stigler) model fitted to the same dataset.
Hence for an equal-abilities, round-robin tournament, the log-influence
weight is an asymptotically efficient estimator for the Bradley--Terry
model.

Now consider a different tournament structure where players hold hands
in a circle. The corresponding citation matrix has a \emph{cycle} or
\emph{circumplex} structure: that is, a band with non-zero entries on
the sub-diagonal and super-diagonal and in the top-right and bottom-left
corners. This might also be described as a \emph{circulant} matrix,
generated by the vector

\[
{   c} =\begin{pmatrix}0 & k & 0 & 0 & \cdots & 0 & k\end{pmatrix}^T
\]

in the first column. The citation matrix for an \(n\)-player circular
tournament, where every player cites each neighbour \(k\) times, would
look like

\[
{   C} = \begin{bmatrix}  & k &        &        &   & k \\k &   &    k &        &   &   \\  & k &          & \ddots &   &   \\  &   & \ddots   &        & k &   \\  &   &          &   k    &   & k \\k &   &          &        & k &  \end{bmatrix}. 
\]

The corresponding probability transition matrix is the same structure
with every \(k\) replaced by \(1/2\), and clearly the unperturbed
PageRank and influence weight vectors are both equal to \(\frac1n{ e}\).

Using the same approach as for the round robin tournament, we find the
variances of log-influence weights (i.e.~diagonal entries of the
covariance matrix) for the circular tournament are equal to
\((n^2 - 1)/(6kn)\). The next covariance terms (i.e.~the
super/sub-diagonal entries) are \((n-1)(n-5)/(6kn)\) followed by (in the
third band, assuming \(n\geq 3\)), \((n^2 - 12n + 23)/(6kn)\) and so on.
These are exactly equal to the asymptotic covariances for a
Bradley--Terry model fit to the same data. So efficiency holds for
circular as well as round robin tournaments.

\hypertarget{discussion}{%
\section{Discussion}\label{discussion}}

Just as the intuition behind PageRank has been discovered and
rediscovered over the years, the connection between eigenvector
centrality and the Bradley--Terry model has also been identified before.

Daniels (\protect\hyperlink{ref-daniels1969}{1969}) proposed a `fair
scores' model for paired comparisons, equivalent to the influence weight
metric of Pinski and Narin (\protect\hyperlink{ref-pinski1976}{1976})
and noted a relationship with Bradley--Terry model. Negahban, Oh, and
Shah (\protect\hyperlink{ref-negahban2017}{2017}) introduced a measure
called Rank Centrality, akin to influence weight for ratio matrices,
which corresponds to Bradley--Terry when every entity competes the same
number of times; their results are equivalent to those in
Section~\ref{sec-asymptotic}. Maystre and Grossglauser
(\protect\hyperlink{ref-maystre2015}{2015}) presented a metric called
Luce Spectral Ranking, equal to influence weight, and an iterative
algorithm for accurate estimation of Bradley--Terry maximum likelihood
estimates from a generalization of the corresponding stationary
distribution. More recently Prathap, Mukherjee, and Leydesdorff
(\protect\hyperlink{ref-prathap2020}{2020}) highlighted the invariance
of influence weight to self-citations (a property mentioned in the
original paper).

\hypertarget{self-citation}{%
\paragraph{Self citation}\label{self-citation}}

In directed graphs such as citation networks, entities have some control
over their out-degree (who they cite) and loops (self-citations),
notions not applicable to sports tournaments, where teams can neither
play against themselves nor choose their opponents.

The journal impact factor is notorious for its sensitivity to
self-citations and its resulting effect on researcher behaviour
(\protect\hyperlink{ref-wilhite2012}{Wilhite and Fong 2012}). Journal
PageRank (\protect\hyperlink{ref-bollen2006}{Bollen, Rodriquez, and Van
De Sompel 2006}) does not discount self-citations, Eigenfactor
(\protect\hyperlink{ref-west2010}{West 2010}) explicitly sets the
diagonal of \(C\) to zero, whilst SJR
(\protect\hyperlink{ref-gonzuxe1lez-pereira2010}{González-Pereira,
Guerrero-Bote, and Moya-Anegón 2010}) limits their number to a third of
a journal's total references. Meanwhile, the export scores model and
influence weight are implicitly invariant to self-citations, so such
adjustments are unnecessary.

\hypertarget{size-bias}{%
\paragraph{Size bias}\label{size-bias}}

Eigenfactor and Article Influence score involve a personalized
teleportation vector based on the number of articles in each journal,
adding an explicit journal size bias. PageRank more generally measures
the `total influence', so a particularly prolific journal including many
articles is likely to have more influence simply by having more things
in it to be cited, regardless of article quality. While `influence per
publication' metrics, including Article Influence score, attempts to
account for this by dividing by publication counts, this does not deal
with the `review journal problem' concerning very long articles
containing more `citable content', nor can it tackle `link farm'-esque
journals who issue many citations in the hope of them being
reciprocated. In social networks this might correspond to a user who
`follows' or `likes' an inordinate number of others. The invariant
method of Palacios-Huerta and Volij
(\protect\hyperlink{ref-palacios-huerta2002}{2002}) accounts for
`reference intensity' but this requires publication count data, which
may not be available, and limits its application to bibliometrics where
this concept exists.

Influence weight, or to use its longer name \emph{influence per outgoing
reference}, explicitly penalizes journals for giving out a lot of
citations and rewards them for being stingy, thus taking away the
`review journal advantage' and automatically discincentivizing `citation
cartels' or `mega journals'. This property motivates the nickname
\textbf{Scroogefactor} for an eigenvector-based citation metric that
disincentivizes such behaviour, tackling two weaknesses of impact factor
and Eigenfactor.

\hypertarget{statistical-analysis}{%
\paragraph{Statistical analysis}\label{statistical-analysis}}

Bibliometrics is sometimes defined as the `statistical analysis of
written publications', but methods are arguably not particularly
statistical without quantification of uncertainty and lack of fit. The
Scroogefactor, through its connection with the Bradley--Terry model,
offers these features, whereas other measures of influence do not.

\hypertarget{limitations-and-further-work}{%
\paragraph{Limitations and further
work}\label{limitations-and-further-work}}

In this paper, we have rather liberally used `PageRank' to refer to
eigenvector centrality even without any damping factor, because the
existence of teleportation breaks the quasi-symmetry property. The
question remains, to what extent this departure can be quantified and if
this might be propagated through derivation of analogous regularization
measures for Bradley--Terry models, such as those based on
`pseudocounts' that preserve quasi-symmetry. This is left for future
work.

\hypertarget{summary}{%
\paragraph{Summary}\label{summary}}

The proof presented here, motivated by quasi-symmetry, is as far as we
can tell, novel. Moreover, the recontextualization of influence weight
as a `scaled PageRank' motivates a different way of thinking about
measures of influence.

\hypertarget{references}{%
\section*{References}\label{references}}
\addcontentsline{toc}{section}{References}

\hypertarget{refs}{}
\begin{CSLReferences}{1}{0}
\leavevmode\vadjust pre{\hypertarget{ref-agresti2013}{}}%
Agresti, Alan. 2013. \emph{Categorical Data Analysis}. Third. Wiley
Series in Probability and Statistics. New York: John Wiley \& Sons.

\leavevmode\vadjust pre{\hypertarget{ref-avrachenkov2010}{}}%
Avrachenkov, Konstantin, Bruno Ribeiro, and Don Towsley. 2010.
{`Improving Random Walk Estimation Accuracy with Uniform Restarts'}. In,
edited by Ravi Kumar and Dandapani Sivakumar, 6516:98--109. Berlin,
Heidelberg: Springer Berlin Heidelberg.
\url{http://link.springer.com/10.1007/978-3-642-18009-5_10}.

\leavevmode\vadjust pre{\hypertarget{ref-bergstrom2007}{}}%
Bergstrom, Carl. 2007. {`Eigenfactor: Measuring the Value and Prestige
of Scholarly Journals'}. \emph{College \& Research Libraries News} 68
(5): 314--16. \url{https://doi.org/10.5860/crln.68.5.7804}.

\leavevmode\vadjust pre{\hypertarget{ref-bof2017}{}}%
Bof, Nicoletta, Giacomo Baggio, and Sandro Zampieri. 2017. {`On the Role
of Network Centrality in the Controllability of Complex Networks'}.
\emph{IEEE Transactions on Control of Network Systems} 4 (3): 643653.
\url{https://doi.org/10.1109/tcns.2016.2550862}.

\leavevmode\vadjust pre{\hypertarget{ref-bollen2006}{}}%
Bollen, Johan, Marko A. Rodriquez, and Herbert Van De Sompel. 2006.
{`Journal Status'}. \emph{Scientometrics} 69 (3): 669--87.
\url{https://doi.org/10.1007/s11192-006-0176-z}.

\leavevmode\vadjust pre{\hypertarget{ref-bradley1952}{}}%
Bradley, Ralph Allan, and Milton E. Terry. 1952. {`Rank Analysis of
Incomplete Block Designs: I. The Method of Paired Comparisons'}.
\emph{Biometrika} 39 (3/4): 324. \url{https://doi.org/10.2307/2334029}.

\leavevmode\vadjust pre{\hypertarget{ref-caussinus1965}{}}%
Caussinus, Henri. 1965. {`Contribution à l'analyse Statistique Des
Tableaux de Corrélation'}. \emph{Annales de La Faculté Des Sciences de
Toulouse Mathématiques} 29: 77--183.
\url{https://doi.org/10.5802/afst.519}.

\leavevmode\vadjust pre{\hypertarget{ref-daniels1969}{}}%
Daniels, H. E. 1969. {`Round-Robin Tournament Scores'}.
\emph{Biometrika} 56 (2): 295--99.
\url{https://doi.org/10.1093/biomet/56.2.295}.

\leavevmode\vadjust pre{\hypertarget{ref-dias2016}{}}%
Dias, Elisângela Silva, Diane Castonguay, and Mitre Costa Dourado. 2016.
{`Algorithms and Properties for Positive Symmetrizable Matrices'}.
\emph{TEMA Tend. Mat. Apl. Comput.} 17 (2): 187.

\leavevmode\vadjust pre{\hypertarget{ref-falagas2008}{}}%
Falagas, Matthew E., Vasilios D. Kouranos, Ricardo Arencibia-Jorge, and
Drosos E. Karageorgopoulos. 2008. {`Comparison of SCImago Journal Rank
Indicator with Journal Impact Factor'}. \emph{The FASEB Journal} 22 (8):
2623--28. \url{https://doi.org/10.1096/fj.08-107938}.

\leavevmode\vadjust pre{\hypertarget{ref-firth2004}{}}%
Firth, D. 2004. {`Quasi-Variances'}. \emph{Biometrika} 91 (1): 65--80.
\url{https://doi.org/10.1093/biomet/91.1.65}.

\leavevmode\vadjust pre{\hypertarget{ref-franceschet2011}{}}%
Franceschet, Massimo. 2011. {`PageRank: Standing on the Shoulders of
Giants'}. \emph{Communications of the ACM} 54 (6): 92--101.
\url{https://doi.org/10.1145/1953122.1953146}.

\leavevmode\vadjust pre{\hypertarget{ref-geller1978}{}}%
Geller, Nancy L. 1978. {`On the Citation Influence Methodology of Pinski
and Narin'}. \emph{Information Processing \& Management} 14 (2): 9395.

\leavevmode\vadjust pre{\hypertarget{ref-golub1986}{}}%
Golub, Gene H., and Carl D. Meyer, Jr. 1986. {`Using the QR
Factorization and Group Inversion to Compute, Differentiate, and
Estimate the Sensitivity of Stationary Probabilities for Markov
Chains'}. \emph{SIAM Journal on Algebraic Discrete Methods} 7 (2):
273--81. \url{https://doi.org/10.1137/0607031}.

\leavevmode\vadjust pre{\hypertarget{ref-gonzuxe1lez-pereira2010}{}}%
González-Pereira, Borja, Vicente P. Guerrero-Bote, and Félix
Moya-Anegón. 2010. {`A New Approach to the Metric of Journals{'}
Scientific Prestige: The SJR Indicator'}. \emph{Journal of Informetrics}
4 (3): 379--91. \url{https://doi.org/10.1016/j.joi.2010.03.002}.

\leavevmode\vadjust pre{\hypertarget{ref-katz1953}{}}%
Katz, Leo. 1953. {`A New Status Index Derived from Sociometric
Analysis'}. \emph{Psychometrika} 18 (1): 39--43.
\url{https://doi.org/10.1007/BF02289026}.

\leavevmode\vadjust pre{\hypertarget{ref-kelly1979}{}}%
Kelly, Frank P. 1979. \emph{Reversibility and Stochastic Networks}.
Wiley.

\leavevmode\vadjust pre{\hypertarget{ref-maystre2015}{}}%
Maystre, Lucas, and Matthias Grossglauser. 2015. {`Fast and Accurate
Inference of Plackett{\textendash}luce Models'}. In, 172180. NIPS'15.
Cambridge, MA, USA: MIT Press.

\leavevmode\vadjust pre{\hypertarget{ref-mccullagh1982}{}}%
McCullagh, Peter. 1982. {`Some Applications of Quasisymmetry'}.
\emph{Biometrika} 69 (2): 303308.
\url{https://doi.org/10.1093/biomet/69.2.303}.

\leavevmode\vadjust pre{\hypertarget{ref-meyer2000}{}}%
Meyer, Carl. 2000. {`Matrix Analysis and Applied Linear Algebra'},
January. \url{https://doi.org/10.1137/1.9780898719512}.

\leavevmode\vadjust pre{\hypertarget{ref-negahban2017}{}}%
Negahban, Sahand, Sewoong Oh, and Devavrat Shah. 2017. {`{R}ank
{C}entrality: Ranking from Pairwise Comparisons'}. \emph{Operations
Research} 65 (1): 266287.

\leavevmode\vadjust pre{\hypertarget{ref-newman2010}{}}%
Newman, Mark. 2010. \emph{Networks: An Introduction}. Oxford University
Press.

\leavevmode\vadjust pre{\hypertarget{ref-nielsen2015}{}}%
Nielsen, A. J. N., and M. Weber. 2015. {`Computing the Nearest
Reversible Markov Chain'}. \emph{Numerical Linear Algebra with
Applications} 22 (3): 483--99. \url{https://doi.org/10.1002/nla.1967}.

\leavevmode\vadjust pre{\hypertarget{ref-page1999}{}}%
Page, Lawrence, Sergey Brin, Rajeev Motwani, and Terry Winograd. 1999.
{`The PageRank Citation Ranking: Bringing Order to the Web.'} Technical
Report 1999-66. Stanford InfoLab; Stanford InfoLab.
\url{http://ilpubs.stanford.edu:8090/422/}.

\leavevmode\vadjust pre{\hypertarget{ref-palacios-huerta2002}{}}%
Palacios-Huerta, Ignacio, and Oscar Volij. 2002. {`The Measurement of
Intellectual Influence'}. \emph{SSRN Electronic Journal}.
\url{https://doi.org/10.2139/ssrn.329803}.

\leavevmode\vadjust pre{\hypertarget{ref-pinski1976}{}}%
Pinski, Gabriel, and Francis Narin. 1976. {`Citation Influence for
Journal Aggregates of Scientific Publications: Theory, with Application
to the Literature of Physics'}. \emph{Information Processing \&
Management} 12 (5): 297312.
\url{https://doi.org/10.1016/0306-4573(76)90048-0}.

\leavevmode\vadjust pre{\hypertarget{ref-prathap2020}{}}%
Prathap, Gangan, Somenath Mukherjee, and Loet Leydesdorff. 2020.
{`Within-Journal Self-Citations and the Pinski-Narin Influence
Weights'}. \emph{Journal of Informetrics} 14 (1): 100989.
\url{https://doi.org/10.1016/j.joi.2019.100989}.

\leavevmode\vadjust pre{\hypertarget{ref-selby2020}{}}%
Selby, David Antony. 2020. {`Statistical Modelling of Citation Networks,
Research Influence and Journal Prestige'}. PhD thesis, Coventry, UK:
University of Warwick. \url{http://wrap.warwick.ac.uk/158562/}.

\leavevmode\vadjust pre{\hypertarget{ref-sharp2000}{}}%
Sharp, WE, and Thomas Markham. 2000. {`Quasi-Symmetry and Reversible
Markov Sequences in Sedimentary Sections'}. \emph{Mathematical Geology}
32 (5): 561579.

\leavevmode\vadjust pre{\hypertarget{ref-stigler1994}{}}%
Stigler, Stephen M. 1994. {`Citation Patterns in the Journals of
Statistics and Probability'}. \emph{Statistical Science} 9 (1): 94--108.
\url{https://www.jstor.org/stable/2246292}.

\leavevmode\vadjust pre{\hypertarget{ref-tahata2022}{}}%
Tahata, Kouji. 2022. {`Advances in Quasi-Symmetry for Square Contingency
Tables'}. \emph{Symmetry} 14 (5): 1051.
\url{https://doi.org/10.3390/sym14051051}.

\leavevmode\vadjust pre{\hypertarget{ref-turner2012}{}}%
Turner, Heather, and David Firth. 2012. {`Bradley-Terry Models in
{\emph{R}} : The {\textbf{BradleyTerry2}} Package'}. \emph{Journal of
Statistical Software} 48 (9).
\url{https://doi.org/10.18637/jss.v048.i09}.

\leavevmode\vadjust pre{\hypertarget{ref-varin2016}{}}%
Varin, Cristiano, Manuela Cattelan, and David Firth. 2016. {`Statistical
Modelling of Citation Exchange Between Statistics Journals'}.
\emph{Journal of the Royal Statistical Society Series A: Statistics in
Society} 179 (1): 1--63. \url{https://doi.org/10.1111/rssa.12124}.

\leavevmode\vadjust pre{\hypertarget{ref-vigna2016}{}}%
Vigna, Sebastiano. 2016. {`Spectral Ranking'}. \emph{Network Science} 4
(4): 433--45. \url{https://doi.org/10.1017/nws.2016.21}.

\leavevmode\vadjust pre{\hypertarget{ref-west2010}{}}%
West, Jevin D. 2010. {`Eigenfactor: Ranking and Mapping Scientific
Knowledge'}. PhD thesis, University of Washington.

\leavevmode\vadjust pre{\hypertarget{ref-wilhite2012}{}}%
Wilhite, Allen W., and Eric A. Fong. 2012. {`Coercive Citation in
Academic Publishing'}. \emph{Science} 335 (6068): 542--43.
\url{https://doi.org/10.1126/science.1212540}.

\end{CSLReferences}

\end{document}